\documentclass[10pt,leqno]{amsart}
\usepackage{graphicx}
\usepackage{indentfirst,csquotes}

\usepackage[indentafter]{titlesec}
\titleformat{name=\section}{}{\thetitle.}{0.8em}{\centering\scshape}
\usepackage{amsfonts}
\newtheorem{theorem}{Theorem}

\theoremstyle{definition}
\newtheorem{definition}{Definition}

\usepackage{amssymb,amsthm,amsmath}
\usepackage{xcolor,paralist,hyperref,titlesec,fancyhdr,etoolbox}

\titleformat{\section}[display]{\normalfont\huge\bfseries\centering}{\centering\chaptertitlename\thechapter}{10pt}{\Large}
\titlespacing*{\section}{0pt}{0ex}{0ex}

\hypersetup{ colorlinks=true, linkcolor=black, filecolor=black, urlcolor=black }

\usepackage{lipsum}

\begin{document}
\title{Upper and Lower Bounds on $T_1$ and $T_2$ decision tree model} %%%%%%%%%%%%
\author{Yousef M. Alhamdan}
\date{\today}

\let\thefootnote\relax
%\address{Mecca, Saudi Arabia}
%\email{ymhamdan@uqu.edu.sa}

\footnotetext{Umm AlQura University}
\footnote{ymhamdan@uqu.edu.sa}

\begin{abstract}
We study a decision tree model in which one is allowed to query subsets of variables. This model is a generalization of the standard decision tree model. For example, the $\lor-$decision (or $T_1$-decision) model has two queries, one is a bit-query and one is the $\lor$-query with arbitrary variables. We show that a monotone property graph, i.e. nontree graph is lower bounded by $n\log n$ in $T_1$-decision tree model. Also, in a different but stronger model, $T_2$-decision tree model, we show that the majority function and symmetric function can be queried in $\frac{3n}{4}$ and $n$, respectively.
\end{abstract} %%%%%%%%%
\maketitle

\paragraph{\textbf{Introduction.}}
In the following, we follow the same notation and definition as in [2] and [3]. A decision tree T is a rooted directed binary tree.  Each of its leaves is labeled 0 or 1. Each internal vertex $v$ is labeled with some function $q_v: \{0, 1\} \rightarrow \{0,1\}$.  Each internal node has two outgoing edges, one labeled $1$ and one labeled by $0$.  A computation of T on input $x \in \{0,1\}^n$ is the path from the root to one of the leaves that in each of the internal vertices $v$ follows the edge and has a label equal to the value of $q_v(x)$. The label of the leaf reached by the path is the output of the computation.  The tree T computes the function $f:\{0,1\} \rightarrow \{0,1\}$ iff on each input $x \in \{0,1 \}^n$ the output of T is equal to $f(x)$.

There are different models of decision tree, each different by the types of function $q_v$ that are allowed in the interval vertices.  For any set $Q$ of functions the decision tree complexity of the function f is the minimal depth of a tree, (i.e. the number of edges in the longest path from the root to a leaf) using functions from Q and computing f.  We denote this value by $D_Q(f)$. The standard decision tree model allows to query individual variables in the vertices of the tree. The complexity in this model is simply denoted by $D(f)$.  A graph property is a set of graphs $S$ such that if the set contains some graph G, i.e. $G \subseteq S$, then it contains each isomorphic copy of G, i.e. $\pi (G) \in S$. Our goal in decision tree complexity is to decide whether a graph has the property P or not(e.g. P could be connectedness, Bipartiteness, etc.). Consider two graphs $G_1$ and $G_2$ on vertices set $\{ 1, 2, 3, ..., n \}$. Denote by $E(G_i)$ the edge set of graph $G_i$, where $i=1,2$. A graph property P is monotone increasing if $G_1 \in P$ and $E(G_1) \subseteq E(G_2)$ imply $G_2 \in P$. P is monotone decreasing if $G_2 \in P$ and $E(G_1) \subseteq E(G_2)$ imply $G_1 \in P$. P is monotone if it is either monotone increasing or monotone decreasing. P is called to be trivial if $P = \phi$ or $P$ contains all graphs on n vertices.

\newpage
\paragraph{\textbf{$T_1-$decision Tree Model.}}

\begin{definition}
Let $I \subseteq {X \choose k-1}$ where $X=\{x_1, x_2, ..., x_n\}$ and $|X|=n$.
\end{definition}

Let $u \in I$, then $u'$ is the string bits of u such that $u'(i)=1$ if and only if $x_i \in u$ where $u'=u(1)u(2)...u(n)$ and $1 \leq i \leq n$ and $n$ is the length of the input. Observe that $\forall u \in I, T_k(u')=0$. This means that for all $u \in I$, there exists a path from root to 0-output leaf.

Given an n-vertex graph $G$ as input, associate a boolean variable $x_e$ with each possible pair $e = \{ u, v \}$ of vertices $u \neq v$. Then, each assignment x of $\binom{n}{2}$ boolean values to these variables gives us a graph $G_x$ with the edge-set $E = \{ e | x_e = 1 \}$. 
\begin{definition}
$g(x) = 1$ if and only if $G_x$ is a non-tree.
\end{definition}
This problem is monotone, it is interested that in $\lor$-decision tree model (more generalization and powerful of the standard model), there exists some inputs that needs $\Omega(n \log n)$ queries to find whether it is non-tree or not. It is still not obvious if Aanderaa–Karp–Rosenberg conjecture applies in the $\lor$ model or not.

\begin{theorem}
$n\log(n-2) \leq D_{\lor} (g)$
\end{theorem}
\begin{proof}
Let $L$ be the set of all labeled trees. Then the number of labeled trees is $(n)^{n-2}$ by Cayley's formula [4], i.e. $|L|=(n)^{n-2}$.

Now, let $A, B \in L$ such that $A \neq B$, then we know that $g(A')=g(B')=0$. Now, taking the union of $A$ and $B$ will result a cycle, i.e. $(A \cup B)$ has a cycle, therefore $g((A \cup B))=1$. In order to prove this claim,  let $T, T' \in L$ such that $T \neq T'$. Now, $T$ is different from $T'$ by at least one edge. If we take the union of $T$ and $T'$, then the resulted graph must have a cycle. This is because $T$ and $T'$ both are tree and $T \neq T'$, then $(T \cup T')$ must have a cycle. Thus, we have as many as $(n)^{n-2}$ different trees in L and all of them have different path from the root to the 0-output. Thus, the depth of the tree is at least $n\log(n-2)$.
\end{proof}

The previous theorem used an old technique which proved a lower bound for threshold function $T_k$ in $\lor$-decision model by Ben-Asher and Newman [1].
\newline
\newline
\newline
\newline
\newline
\paragraph{\textbf{$T_2-$decision Tree Model}}
The following theorem is an interesting in its own because majority function can be queried using only one kind of queries, i.e. $T_2$. Benasher and Newman [1] show that using only $\lor$ (or $T_1$) queries, then we have at most $\frac{n}{2}$ queries (in fact, in the same paper, they also show that this is optimal). If we add another function to the model as Chistopolskaya and Podolskii did in their paper [2], they show that in $\{ \lor, \land \}$-decision tree model, they show that $n-1$ is optimal for majority function. In this paper, we study majority function in $T_2$-decision tree model, we show that majority function can be queried in at most $\frac{3n}{4}$.

\begin{theorem}
    $D_{T_2^n}(MAJ_n) \leq \frac{3n}{4} $
\end{theorem}

\begin{proof}
    Split all variables into blocks of size 2, i.e. $(x_1, x_2), (x_3, x_4), ..., (x_{n-1}, x_n)$. There are two levels of queries:

    \begin{itemize}
        \item [1.] Query each block by $T_2$, let $L_0^1$ and $L_1^1$ be the number of blocks that answers 0 and 1 in level 1, respectively. If $L_1^1 \geq \frac{n}{4}$, then answer 1. Otherwise, go to level 2.
        \item [2. ] Group every two blocks of size two that answer 0 into one bigger block of size 4 (e.g. let's say $T_2(x_1, x_2) = 0$ and $T_2(x_3, x_4) = 0$, then we will group $x_1, x_2, x_3, x_4$ into one block.). Query each block by $T_2$ and let $L_0^2$ and $L_1^2$ be the number of blocks that answers 0 and 1 in level 2, respectively. If $L_1^1 + L_1^2 \geq \frac{n}{4}$, then answer 1, otherwise answer 0.
    \end{itemize}

    Now we shall prove why these two levels are true. For level 1 and 2, observe that if $L_1^1 \geq \frac{n}{4}$ or $L_1^1 + L_1^2 \geq \frac{n}{4}$, then each block has exactly two 1s, therefore there are $2*\frac{n}{4} = \frac{n}{2}$ 1s. Note that we don't count bits twice, i.e. the input bits in $L_1^1$ have no common bits in input bits of $L_1^2$; this is because at level 2, we only query pairs that answers 0. Otherwise at level 1, $L_0^1 \geq \frac{n+4}{4}$ and at level 2, $L_0^2 \geq \lfloor \frac{\frac{n}{2}+1}{4} \rfloor = \lfloor \frac{n+2}{8} \rfloor$. In this case, the number of 0s in the input will be at least $(\frac{n+4}{4})*1 + (\frac{n+2}{8} )*3 = \frac{5n+14}{8} \geq \frac{n}{2}$  because at level 1 and 2, the number of 0s in each block that answers 0 at least 1 and 3, respectively. Now, observe that the total number of queries are $\frac{n}{2} + \frac{n}{4} = \frac{3n}{4}$.

\end{proof}

The following theorem shows that for any symmetric function e.g. $T_k$ for $k\geq 2$, $PARITY$ can be queries in at most n queries. 
\begin{theorem}\
    $D_{T_2^n}(SYM_n) \leq n $
\end{theorem}
\begin{proof}
    Split all variables into blocks of size 4, let us denote one of the blocks by $x_1, x_2, x_3, x_4$. Now, each block will have at most 4 queries. Let $y_1 = T_2(x_1, x_2)$, $y_2 = T_2(x_3, x_4)$, $y_3 = T_2(x_1, x_2, x_3, x_4)$, $y_4 = T_2(x_1, x_2, x_3, x_4, x_1, x_2, x_3, x_4)$, $y_5 = T_2(x_1, x_3)$ and $y_6 = T_2(x_2, x_4)$. Then, We have the following cases: Query $y_1$ and $y_2$
\begin{itemize}
    
    \item [case 1:] if both $y_1$ and $y_2$ equal zero, then this block has either two, one or zero ones. In order to differentiate between these cases, query $y_3$. If $y_3=1$, then there is exactly two ones. Otherwise, query $y_4$, if $y_4=0$, then all variables in this block equal zero. Otherwise, there is only one variable equal one.
    
    \item [case 2:] if $y_1 != y_2$ (i.e. one is zero, the other is one), then we know that there is either two ones or three ones. So, we query $y_5$ and $y_6$, if all equal zero, then there are exactly two ones, Otherwise there are exactly three ones.
\end{itemize}
\end{proof}

%\bibliographystyle
%\printbibliography

\end{document}